\documentclass[12pt,a4paper]{article}
\usepackage[T1]{fontenc}
\usepackage[utf8]{inputenc}
\usepackage{lmodern,mathptmx,amsmath,amssymb,amsthm,mathtools,graphicx,booktabs}
\usepackage[a4paper,margin=2cm]{geometry}
\usepackage[numbers,square]{natbib}
\usepackage{hyperref}
\usepackage{microtype}
\hypersetup{pdftitle={Quadratically Weighted Exponential--Gamma Mixtures: Properties, Hazard Behavior, and Statistical Inference},pdfauthor={Lalita Singh and Rohit Patna}}
\numberwithin{equation}{section}
\newtheorem{theorem}{Theorem}
\newtheorem{proposition}[theorem]{Proposition}
\newtheorem{corollary}[theorem]{Corollary}

\newcommand{\E}{\mathbb E}
\newcommand{\Prob}{\mathbb P}

\newcommand{\Exp}{\operatorname{Exp}}
\newcommand{\Gam}{\operatorname{Gamma}}
\newcommand{\dd}{\mathrm d}

\title{The LR2 Lifetime Distribution: Rate Crossing, Reliability Properties, and Estimation}
\author{%
  Lalita Singh\\
  Agra College, Agra\\
  lalitasingh062@gmai.com
  \and
  Rohit Patawa\thanks{Corresponding author: mail}\\
  Agra College, Agra\\
  rohitpatawa@gmail.com
}
\date{}
\begin{document}
\maketitle
\begin{abstract}
We study the LR2 one-parameter lifetime distribution obtained by mixing an exponential component of rate $\theta/2$ and a gamma component with shape four and rate $\theta^2$. The exponential mixing probability is $\theta^2/(1+\theta^2)$ for $\theta>0$. We derive the density, survival function, moments, transforms, truncated moments, quantiles, order-statistic distributions, mean residual life, and stress--strength reliability. Identifiability follows from the density at the origin, and strict monotonicity of the mean yields a unique moment estimator. A complete density-turning-point criterion and analytic initial and tail hazard expansions clarify the possible shapes. The component rates cross at one half. The hazard decreases initially for every positive parameter and eventually increases in the gamma-dominated regime; it eventually decreases in the exponential-dominated regime. We develop likelihood estimation, observed information, constrained expectation--maximization, spacing estimation, and inference under regularity conditions, and state the likelihood for independent right censoring. A reproducible Monte Carlo study with eight settings and 4,000 generated samples compares likelihood and moment estimation and evaluates log-scale confidence intervals. The analysis distinguishes mathematical properties from empirical validation and explains the restrictions associated with a fixed working time unit.
\end{abstract}

\section{Introduction}
The statistical analysis of positive continuous observations is central to the study of lifetimes, waiting times, and durations. In reliability engineering, a distribution describes the operating time of a component before failure; in survival analysis, it describes the time until an event of interest. In both settings, quantities such as the survival probability, failure rate, and expected remaining lifetime provide information that cannot be summarized by the mean alone. The choice of a probability model therefore affects both the interpretation of a duration mechanism and the uncertainty attached to estimated reliability measures. The survival-analysis framework of \citet{cox1984} provides the setting for these distributional and inferential questions.

A useful lifetime model should combine an interpretable structure with enough flexibility to represent the features relevant to the application. The exponential distribution is attractive because of its simple likelihood, memoryless property, and constant hazard rate. This simplicity can become restrictive when the instantaneous failure rate changes with age. The gamma distribution provides a different mechanism: with a positive integer shape, it can describe the sum of independent exponential waiting times with a common rate. Such a stage-based construction produces a duration distribution whose behavior near the origin differs from that of an exponential lifetime. Combining these two components allows both mechanisms to contribute to a population distribution while retaining the positive support and tractable component formulas.

Finite mixture models offer a general way to represent heterogeneity through a weighted combination of component distributions. The component indicators may correspond to unobserved population groups or may serve as a mathematical representation of a more complicated distributional shape. In either interpretation, the observed density is obtained by averaging over the latent component assignment. As reviewed by \citet{mclachlan2019}, mixture models have a well-developed theory and a broad range of statistical uses. Their flexibility comes with inferential questions concerning parameter identification, likelihood optimization, and the interpretation of the fitted components. A parsimonious mixture is consequently useful only when its constraints are understood alongside its distributional properties.

Exponential--gamma constructions have an established place in the literature on lifetime distributions. The Lindley distribution, originating in \citet{lindley1958}, provides a classical one-parameter example based on exponential and shape-two gamma components with a common rate. \citet{ghitany2008} developed its mathematical and reliability properties, including moments, failure rate, and mean residual life. Subsequent work has considered several routes to extending this construction. For example, \citet{ghitany2013} introduced the two-parameter power Lindley distribution and studied its distributional characteristics, likelihood estimation, and finite-sample performance. These developments illustrate how a tractable baseline model can be modified to produce additional shape variation and associated inferential procedures.

Related finite mixtures use other gamma shapes and parameter-dependent mixing probabilities. \citet{sen2016} introduced the xgamma distribution as a particular mixture of an exponential distribution and a shape-three gamma distribution, and investigated moment and likelihood estimation together with survival properties. The Sujatha distribution of \citet{shanker2016} provides another one-parameter lifetime construction based on exponential and gamma components. These models demonstrate that changing the component structure and the mixing rule can yield distinct distributions even when the underlying building blocks are familiar. Their properties also make clear the importance of analyzing the resulting mixture directly: the shape of a component density or hazard cannot, by itself, determine the shape of the population mixture.

Allowing different component rates provides another way to distinguish duration mechanisms. In a closely related study, \citet{sakthivel2023} considered an exponential--gamma mixture with a shape-three gamma component and separate component-rate parameters. An unrestricted specification permits the mixture probability and the two rates to vary independently. A constrained specification instead links these quantities, reducing the number of parameters but imposing relationships among the component frequencies and durations. The relevant question is therefore not simply whether a mixture can be formed, but what distributional and inferential consequences follow from a particular set of constraints.

This paper studies the LR2 restriction with exponential rate $a(\theta)=\theta/2$, gamma rate $b(\theta)=\theta^2$, and exponential mixing probability $w(\theta)=\theta^2/(1+\theta^2)$. The parameter links mixture composition to two different rate functions. The exponential mean is $2/\theta$, while the gamma mean is $4/\theta^2$. We retain ``LR2'' as a working distribution label; it is not an abbreviation for a likelihood ratio.

The distinction between the two rate functions creates a crossing at $\theta=1/2$. Below this value the gamma rate is smaller and its survival component controls the far tail. Above it the exponential component controls the far tail. At equality the gamma component still dominates the tail because its survival includes a cubic polynomial factor. Crossing of the rates should not be confused with crossing of the component means, which occurs at $\theta=2$. A shape-four gamma component can represent four sequential stages, but that interpretation requires application-specific evidence.

The hazard illustrates why mixture structure must be analyzed directly. It decreases near the origin for every positive parameter. For $0<\theta\leq1/2$ it eventually increases and is therefore nonmonotone; for $\theta>1/2$ it eventually decreases toward $\theta/2$. The latter endpoint result does not establish global decrease on the entire support. A separate density threshold determines whether an interior local peak exists. Together these results distinguish density shape, initial failure behavior, and tail failure behavior.

The parameter restriction also calls for careful estimation. A scalar parameterization reduces optimization dimension but does not automatically imply a unique likelihood maximizer. We prove identifiability and strict decrease of the mean, giving a unique moment estimator for every positive sample mean. We derive the likelihood score and observed information without treating unobserved component allocations as known. A constrained expectation--maximization scheme follows the framework of \citet{dempster1977}, while maximum product of spacings follows \citet{cheng1983}. Large-sample likelihood inference is established on compact parameter intervals, with the positivity of the parameter handled through log-scale confidence limits.

The contribution is a self-contained mathematical and computational treatment of this particular restriction of an established mixture construction. Explicit formulas connect the component assumptions to moments and reliability measures; a threshold theorem classifies density stationary points; and hazard expansions identify the behavior near zero and in the tail. A reproducible simulation compares estimation methods under correct specification. The numerical study concerns complete independent generated observations, and no claim of superiority on an observed lifetime dataset is made. Throughout the main text, gamma parameters denote shape and rate, and time is measured in a declared fixed working unit.

The paper is organized as follows. Section~\ref{sec:construction} defines the distribution and establishes validity and identifiability. Section~\ref{sec:moments} develops moments and related measures. Section~\ref{sec:shape} analyzes density and hazard shape, and Section~\ref{sec:reliability} derives reliability functionals. Section~\ref{sec:inference} presents estimation and inference. Section~\ref{sec:numerical} reports the simulation, and Section~\ref{sec:discussion} discusses limitations and extensions. The appendix supplies the exact simulation algorithm and computational details.

\section{Model construction and basic properties}\label{sec:construction}
Throughout, $x>0$ and $\theta>0$. Let $\Exp(a)$ have rate $a$, and let $\Gam(4,b)$ have shape four and rate $b$. Their densities are
\begin{equation}
 e_a(x)=a e^{-ax},\qquad g_b(x)=\frac{b^4x^3}{6}e^{-bx}.
\end{equation}
Define $w=w(\theta)=\theta^2/(1+\theta^2)$ and let $Z$ be Bernoulli with success probability $w$. Conditional on $Z$, take
\begin{equation}\label{eq:latent}
 X\mid Z=1\sim\Exp(a(\theta)),\qquad X\mid Z=0\sim\Gam(4,b(\theta)).
\end{equation}
Thus the common density is
\begin{equation}\label{eq:generalpdf}
 f(x;\theta)=w a e^{-ax}+(1-w)\frac{b^4x^3}{6}e^{-bx}.
\end{equation}
Hereafter, the auxiliary symbols always mean $a=\theta/2$ and $b=\theta^2$; they are not independently estimated parameters. The LR2 density is explicitly
\begin{equation}\label{eq:specificpdf}
 f(x;\theta)=\frac{\theta^3e^{-\theta x/2}/2+\theta^8x^3e^{-\theta^2x}/6}{1+\theta^2},\qquad x>0.
\end{equation}

\begin{proposition}[Validity and survival function]\label{prop:valid}
LR2 is a probability distribution on $(0,\infty)$. Define
\begin{equation}
 P_3(y)=1+y+\frac{y^2}{2}+\frac{y^3}{6}.
\end{equation}
The survival function and distribution function are
\begin{equation}\label{eq:survival}
 S(x;\theta)=w e^{-ax}+(1-w)e^{-bx}P_3(bx),\qquad F(x;\theta)=1-S(x;\theta).
\end{equation}
For $x\leq0$, $F(x;\theta)=0$ and $S(x;\theta)=1$.
\end{proposition}
\begin{proof}
The weights are strictly between zero and one and sum to one. Both component densities integrate to one, proving normalization. Repeated integration by parts gives the gamma survival function $e^{-bx}P_3(bx)$, and conditioning on $Z$ gives \eqref{eq:survival}.
\end{proof}

For this model, the explicit survival function is
\begin{equation}\label{eq:specificsurvival}
 S(x;\theta)=\frac{\theta^2 e^{-\theta x/2}+e^{-\theta^2x}P_3(\theta^2x)}{1+\theta^2}.
\end{equation}
The cumulative hazard is $H(x;\theta)=-\log S(x;\theta)$.

\begin{proposition}[Identifiability within a specification]\label{prop:ident}
For LR2, distinct positive parameters determine distinct distributions.
\end{proposition}
\begin{proof}
The continuous extension of the density to zero is
\begin{equation}\label{eq:origin}
 f(0+;\theta)=wa=\frac{\theta^3}{2(1+\theta^2)},\qquad
 \frac{\dd}{\dd\theta}f(0+;\theta)=\frac{\theta^2(3+\theta^2)}{2(1+\theta^2)^2}>0.
\end{equation}
Equal distributions have equal continuous densities on the positive support and therefore equal limits at zero. Strict monotonicity identifies $\theta$.
\end{proof}

\section{Moments, transforms, and related measures}\label{sec:moments}
\begin{proposition}[Raw moments and transforms]\label{prop:moments}
For every real $r>-1$,
\begin{equation}\label{eq:moments}
 \mu'_r=\E(X^r)=w\frac{\Gamma(r+1)}{a^r}
 +(1-w)\frac{\Gamma(r+4)}{6b^r}.
\end{equation}
The moment-generating function is
\begin{equation}\label{eq:mgf}
 M(t)=w\frac{a}{a-t}+(1-w)\left(\frac{b}{b-t}\right)^4,
 \qquad t<\min(a,b).
\end{equation}
The Laplace transform is $M(-s)$ for $s\geq0$, and the characteristic function is obtained by replacing $t$ by $it$ in \eqref{eq:mgf}.
\end{proposition}
\begin{proof}
Integrate $x^r$ against the two component densities and substitute $y=ax$ and $y=bx$. The exponential integral requires $r>-1$, which also ensures finiteness of the gamma integral. Integration of $e^{tx}$ gives \eqref{eq:mgf}; because both weights are positive, finiteness requires $t$ to be smaller than both rates.
\end{proof}

Writing $\mu=\mu'_1$, the first four raw moments are
\begin{align}
 \mu&=\frac{w}{a}+\frac{4(1-w)}{b},\label{eq:mean}\\
 \mu'_2&=\frac{2w}{a^2}+\frac{20(1-w)}{b^2},\\
 \mu'_3&=\frac{6w}{a^3}+\frac{120(1-w)}{b^3},\\
 \mu'_4&=\frac{24w}{a^4}+\frac{840(1-w)}{b^4}.
\end{align}
In particular, the law of total variance gives the interpretable decomposition
\begin{equation}\label{eq:variance}
 \sigma^2=\frac{w}{a^2}+\frac{4(1-w)}{b^2}
 +w(1-w)\left(\frac1a-\frac4b\right)^2.
\end{equation}
The last term quantifies the contribution of between-component heterogeneity. The coefficient of variation, skewness, and excess kurtosis are
\begin{align}
 \mathrm{CV}&=\sigma/\mu,\\
 \gamma_1&=(\mu'_3-3\mu\mu'_2+2\mu^3)/\sigma^3,\\
 \gamma_2&=(\mu'_4-4\mu\mu'_3+6\mu^2\mu'_2-3\mu^4)/\sigma^4-3.
\end{align}
These expressions determine the dispersion and standardized shape measures of LR2.

For LR2, substitution of the rates yields the fully parameterized raw-moment formula
\begin{equation}\label{eq:specificmoments}
 \mu'_r=\frac{2^r\theta^{2-r}\Gamma(r+1)+\theta^{-2r}\Gamma(r+4)/6}{1+\theta^2},\qquad r>-1.
\end{equation}
In particular,
\begin{equation}\label{eq:specificmeans}
 \mu(\theta)=\frac{2\theta+4/\theta^2}{1+\theta^2},\qquad \mu'_2=\frac{8+20/\theta^4}{1+\theta^2}.
\end{equation}
The variance can equivalently be computed as $\sigma^2=\mu'_2-\mu^2$.
\begin{proposition}[Monotonicity of the mean]\label{prop:meanmono}
The LR2 mean is continuous and strictly decreasing from infinity to zero on $(0,\infty)$.
\end{proposition}
\begin{proof}
Differentiating gives
\begin{equation}
 \frac{\dd\mu(\theta)}{\dd\theta}=\frac{2-2\theta^2-16/\theta-8/\theta^3}{(1+\theta^2)^2}.
\end{equation}
For $\theta\geq1$, $2-2\theta^2\leq0$. For $0<\theta<1$, $2-2\theta^2<2$ while $16/\theta>16$. Thus the numerator is negative throughout. The endpoint limits follow directly from \eqref{eq:specificmeans}.
\end{proof}

\begin{figure}[htbp]\centering
\includegraphics[width=.95\textwidth]{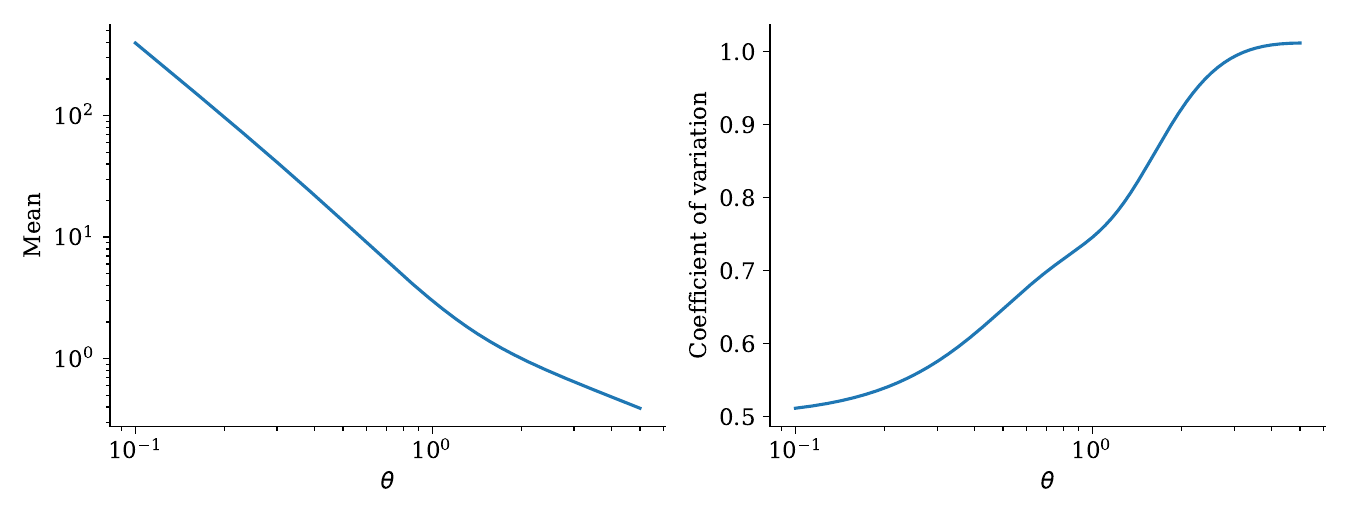}
\caption{Mean and coefficient of variation for LR2. These are direct evaluations of the analytic moment formulas in a fixed working time unit.}\label{fig:moments}
\end{figure}

\subsection{Truncated moments, quantiles, and order statistics}
Let $\gamma(\nu,z)=\int_0^z u^{\nu-1}e^{-u}\dd u$ and $\Gamma(\nu,z)=\int_z^\infty u^{\nu-1}e^{-u}\dd u$ denote the lower and upper incomplete gamma functions. For $r>-1$,
\begin{align}
 T_r(x)&=\E[X^r\mathbf1\{X\leq x\}]
 =w a^{-r}\gamma(r+1,ax)+(1-w)\frac{b^{-r}}6\gamma(r+4,bx),\label{eq:lower}\\
 U_r(x)&=\E[X^r\mathbf1\{X>x\}]
 =w a^{-r}\Gamma(r+1,ax)+(1-w)\frac{b^{-r}}6\Gamma(r+4,bx).\label{eq:upper}
\end{align}
Because $f(x;\theta)>0$, every $p\in(0,1)$ has a unique quantile $Q(p;\theta)$. It can be found by solving $S(Q;\theta)=1-p$ using a bracketed root finder. There is no need to assume an elementary inverse CDF. For a numerically stable bracket, start at zero, choose any positive upper value, and double it until the survival probability is below $1-p$.

For independent observations of size $n$, the $k$th order statistic has density
\begin{equation}
 f_{k:n}(x)=\frac{n!}{(k-1)!(n-k)!}
 F(x)^{k-1}S(x)^{n-k}f(x).
\end{equation}
For example, the minimum has survival $S(x)^n$ and hazard $n h(x)$, corresponding to an independent series system. These identities require no additional distributional assumptions beyond independence.

The Lorenz curve, when this measure is appropriate for a positive quantity, is $L(p)=T_1(Q(p))/\mu$. The mean absolute deviation about the mean is
\begin{equation}
 \E|X-\mu|=2\{\mu F(\mu)-T_1(\mu)\}.
\end{equation}
Both expressions are consequences of the truncated first moment, rather than separate special-purpose derivations.

\section{Density shape and hazard behavior}\label{sec:shape}
\subsection{A complete density-turning-point criterion}
The density need not be unimodal in the sense of having a single peak over the closure of the support. It always decreases immediately to the right of zero, but can acquire an additional interior peak.
\begin{theorem}[Density-shape threshold]\label{thm:mode}
Fix $0<w<1$ and $a,b>0$ in \eqref{eq:generalpdf}. For $0<x<3/b$, define
\begin{equation}\label{eq:Rshape}
 R(x)=\frac{(1-w)b^4}{6wa^2}x^2(3-bx)e^{(a-b)x}.
\end{equation}
This function has a unique maximizer $x_\star\in(0,3/b)$, characterized by
\begin{equation}\label{eq:xstar}
 \frac2{x_\star}-\frac b{3-bx_\star}+(a-b)=0.
\end{equation}
If $R(x_\star)<1$, the density is strictly decreasing. If $R(x_\star)=1$, it remains strictly decreasing but has one horizontal tangent. If $R(x_\star)>1$, the density has exactly two stationary points: a local minimum followed by a local maximum. No additional stationary points occur.
\end{theorem}
\begin{proof}
Differentiation yields
\begin{equation}\label{eq:pdfderivative}
 f'(x)=-wa^2e^{-ax}+(1-w)\frac{b^4}{6}x^2(3-bx)e^{-bx}.
\end{equation}
For $x\geq3/b$, both terms make $f'(x)<0$. On $(0,3/b)$, its sign is the sign of $R(x)-1$. Moreover,
\[
 (\log R)'(x)=\frac2x-\frac b{3-bx}+a-b,\qquad
 (\log R)''(x)=-\frac2{x^2}-\frac{b^2}{(3-bx)^2}<0.
\]
The first derivative decreases continuously from $+\infty$ to $-\infty$. Hence $R$ increases to a unique maximum and then decreases, with zero limits at both interval endpoints. Comparing that maximum with one proves all three cases and the signs on either side of each stationary point.
\end{proof}

For LR2, the dimensionless variable $y=\theta^2x$ gives
\begin{equation}
 R(x)=\frac23 y^2(3-y)\exp\left\{\left(\frac1{2\theta}-1\right)y\right\},\qquad 0<y<3.
\end{equation}
Its unique maximizing $y_\star$ solves $2/y_\star-1/(3-y_\star)+1/(2\theta)-1=0$. Comparing this maximum with one supplies the complete density criterion. At $\theta=1/2$, $a=b=1/4$, $w=1/5$, and $y_\star=2$, giving $R(x_\star)=8/3>1$. Thus equality of rates still gives a boundary peak and an interior local peak; it does not reduce LR2 to a single gamma or exponential distribution.

\subsection{Hazard at the origin and in the tail}
The hazard is
\begin{equation}\label{eq:hazard}
 h(x;\theta)=\frac{wae^{-ax}+(1-w)b^4x^3e^{-bx}/6}
 {we^{-ax}+(1-w)e^{-bx}P_3(bx)}.
\end{equation}
It is also a survival-weighted average of the component hazards:
\begin{equation}\label{eq:hazardmix}
 h(x)=\pi(x)a+\{1-\pi(x)\}h_G(x),\qquad
 \pi(x)=\frac{we^{-ax}}{S(x)},\quad
 h_G(x)=\frac{b^4x^3}{6P_3(bx)}.
\end{equation}
The weights depend on time and differ from the prior mixture weights. This distinction explains why an increasing gamma hazard does not determine the mixture's shape.

\begin{theorem}[Initial and limiting hazard]\label{thm:hazard}
For LR2 and every $\theta>0$,
\begin{equation}\label{eq:hzero}
 h(0+)=wa,\qquad h'(0+)=-w(1-w)a^2<0.
\end{equation}
Furthermore,
\begin{equation}\label{eq:htail}
 \lim_{x\to\infty}h(x)=\min(a,b).
\end{equation}
If $b\leq a$, then
\begin{equation}\label{eq:gammatail}
 h(x)=b-\frac3x+O(x^{-2}),\qquad
 h'(x)=\frac3{x^2}+O(x^{-3}),
\end{equation}
so the hazard eventually increases and cannot be globally monotone. If $a<b$, then
\begin{equation}\label{eq:exptail}
 h(x)-a\sim\frac{1-w}{w}\frac{b^3(b-a)}6x^3e^{-(b-a)x},
\end{equation}
and the hazard is eventually decreasing toward $a$ from above.
\end{theorem}
\begin{proof}
Since $S(0)=1$, $f(0)=wa$, $S'(0)=-wa$, and $f'(0)=-wa^2$, differentiating $h=f/S$ gives \eqref{eq:hzero}. If $b<a$, the gamma survival term dominates exponentially. If $b=a$, its cubic polynomial factor still dominates the exponential component. In either case, expansion of the explicit polynomial ratio in \eqref{eq:hazard} gives \eqref{eq:gammatail}; the exponentially small corrections when $b<a$ and the lower polynomial terms when $b=a$ can be differentiated to give the displayed derivative. If $a<b$, use
\[
 h-a=\frac{(1-w)e^{-bx}\{b^4x^3/6-aP_3(bx)\}}{S(x)}.
\]
The numerator's leading polynomial coefficient is $b^3(b-a)/6$, and $S(x)\sim we^{-ax}$. This proves \eqref{eq:exptail}; differentiation of this explicit rational-exponential expression shows eventual decrease. The common limit follows in all cases.
\end{proof}

\begin{corollary}\label{cor:nonmono}
For LR2, the explicit initial and limiting quantities are
\begin{equation}
 h(0+)=\frac{\theta^3}{2(1+\theta^2)},\qquad
 h'(0+)=-\frac{\theta^4}{4(1+\theta^2)^2},\qquad
 h(\infty)=\min(\theta/2,\theta^2).
\end{equation}
For $0<\theta\leq1/2$, the hazard eventually increases toward $\theta^2$ and is nonmonotone. For $\theta>1/2$, it eventually decreases toward $\theta/2$ from above. Neither statement counts all hazard turning points. No positive parameter yields a globally increasing hazard because the initial derivative is negative.

\end{corollary}
\begin{proof}
Substitute the LR2 rates into Theorem~\ref{thm:hazard}.
\end{proof}
The corollary does not assert exactly one hazard minimum. Counting all hazard turning points requires additional analysis; the density criterion in Theorem~\ref{thm:mode} is a different result.

\begin{figure}[htbp]\centering
\includegraphics[width=\textwidth]{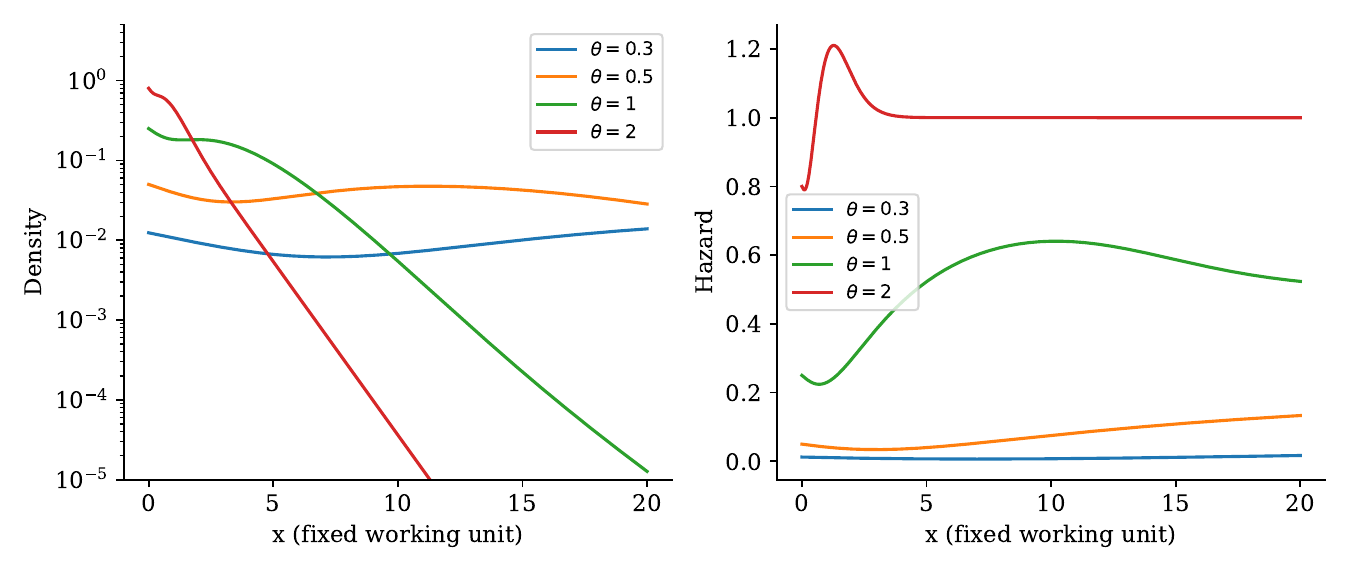}
\caption{LR2 density (left, logarithmic vertical scale) and hazard (right) curves for $\theta=0.3,0.5,1,2$. The finite plotting range illustrates behavior but does not replace the tail results of Theorem~\ref{thm:hazard}.}\label{fig:shapes}
\end{figure}

\section{Reliability functionals}\label{sec:reliability}
\subsection{Mean residual life and mean inactivity time}
The mean residual life at age $x$ is $m(x)=\E(X-x\mid X>x)$. Integrating \eqref{eq:survival} term by term gives
\begin{equation}\label{eq:mrl}
 m(x)=\frac{w e^{-ax}/a+(1-w)e^{-bx}B_3(bx)/b}{S(x)},
 \quad B_3(y)=4+3y+y^2+\frac{y^3}{6}.
\end{equation}
Equivalently, $m(x)=U_1(x)/S(x)-x$. The mean inactivity time is
\begin{equation}
 \widetilde m(x)=\E(x-X\mid X\leq x)=x-\frac{T_1(x)}{F(x)},\qquad x>0.
\end{equation}
These standard lifetime summaries complement the hazard; see \citet{cox1984} for the survival-analysis setting. At zero, $m(0)=\mu$, and its limiting value is $1/\min(a,b)$. The latter follows directly from \eqref{eq:mrl} and the dominant survival component.

\subsection{Stress--strength reliability}
Let $X$ denote strength and $Y$ stress, with independent mixture distributions. Both follow LR2 and may have different parameter values. Write $(w_X,a_X,b_X)$ and $(w_Y,a_Y,b_Y)$ for their component specifications. Define reliability as $\mathcal R=\Prob(Y<X)$; with this convention,
\begin{equation}\label{eq:reliability}
 \mathcal R=w_Xw_Y R_{EE}+w_X(1-w_Y)R_{EG}
 +(1-w_X)w_YR_{GE}+(1-w_X)(1-w_Y)R_{GG},
\end{equation}
where the first subscript refers to the strength component and
\begin{align}
 R_{EE}&=\frac{a_Y}{a_X+a_Y},\\
 R_{EG}&=\left(\frac{b_Y}{b_Y+a_X}\right)^4,\\
 R_{GE}&=1-\left(\frac{b_X}{b_X+a_Y}\right)^4,\\
 R_{GG}&=1-\sum_{k=0}^3\frac{\Gamma(k+4)}{6k!}
 \frac{b_X^4b_Y^k}{(b_X+b_Y)^{k+4}}.\label{eq:RGG}
\end{align}
To verify these expressions, condition on both latent components. For $R_{EG}$, integrate $S_{\Exp(a_X)}(y)$ against the gamma stress density, giving its Laplace transform. For $R_{GE}$, integrate the exponential stress CDF against the gamma strength density. For $R_{GG}$, integrate $S_Y(x)$ against the strength density and subtract from one. Equal distributions yield $\mathcal R=1/2$, as required by independence and continuity.

\section{Estimation and statistical inference}\label{sec:inference}
\subsection{Moment estimation}
Suppose $X_1,\ldots,X_n$ are independent and identically distributed from LR2. The moment estimator $\widehat\theta_{\rm MM}$ solves
\begin{equation}\label{eq:momest}
 \mu(\widehat\theta_{\rm MM})=\overline X.
\end{equation}
Proposition~\ref{prop:meanmono} proves existence and uniqueness for every strictly positive sample mean. Multiplication by the positive denominator gives the model-specific equation
\begin{equation}
 \overline X\theta^4-2\theta^3+\overline X\theta^2-4=0.
\end{equation}
Despite the polynomial representation, a monotone root solver on $\eta=\log\theta$ is more convenient than explicit radical formulas. The law of large numbers and continuity of the inverse mean prove consistency. Since all moments are finite, the central limit theorem and delta method give
\begin{equation}\label{eq:momasymptotic}
 \sqrt n(\widehat\theta_{\rm MM}-\theta_0)
 \xrightarrow{d}N\left(0,\frac{\sigma^2(\theta_0)}{\{\mu'(\theta_0)\}^2}\right).
\end{equation}
The variance in \eqref{eq:momasymptotic} is readily estimated by plugging the moment estimate into the analytic expressions.

\subsection{Likelihood, score, and observed information}
Define the weighted component densities
\begin{equation}
 A_i=wae^{-aX_i},\qquad B_i=(1-w)b^4X_i^3e^{-bX_i}/6.
\end{equation}
The log likelihood and posterior component probability are
\begin{equation}\label{eq:likelihood}
 \ell(\theta)=\sum_{i=1}^n\log(A_i+B_i),\qquad
 r_i(\theta)=\frac{A_i}{A_i+B_i}.
\end{equation}
Let primes on rates denote derivatives with respect to $\theta$. The weighted-component scores are
\begin{align}
 s_{Ei}&=\frac2\theta-\frac{2\theta}{1+\theta^2}+\frac{a'}a-a'X_i,\label{eq:sE}\\
 s_{Gi}&=-\frac{2\theta}{1+\theta^2}+4\frac{b'}b-b'X_i.\label{eq:sG}
\end{align}
Then
\begin{equation}\label{eq:score}
 U(\theta)=\ell'(\theta)=\sum_{i=1}^n\{r_i s_{Ei}+(1-r_i)s_{Gi}\}.
\end{equation}
Writing $c(\theta)=-2(1-\theta^2)/(1+\theta^2)^2$, the score derivatives are
\begin{align}
 s'_{Ei}&=-\frac2{\theta^2}+c(\theta)+\frac{a''}a-\left(\frac{a'}a\right)^2-a''X_i,\\
 s'_{Gi}&=c(\theta)+4\left\{\frac{b''}b-\left(\frac{b'}b\right)^2\right\}-b''X_i.
\end{align}
Since $r_i'=r_i(1-r_i)(s_{Ei}-s_{Gi})$, the observed information is
\begin{equation}\label{eq:info}
 J_n(\theta)=-\sum_{i=1}^n\left\{r_i s'_{Ei}+(1-r_i)s'_{Gi}
 +r_i(1-r_i)(s_{Ei}-s_{Gi})^2\right\}.
\end{equation}
The final variance term is essential; omitting it treats unobserved component assignments as known. For LR2, the rate derivatives are
\begin{equation}
 a'=1/2,\quad a''=0,\quad b'=2\theta,\quad b''=2.
\end{equation}
The component scores therefore simplify to
\begin{equation}
  s_{Ei}=\frac3\theta-\frac{2\theta}{1+\theta^2}-\frac{X_i}{2},\qquad s_{Gi}=\frac8\theta-\frac{2\theta}{1+\theta^2}-2\theta X_i.
\end{equation}
These explicit expressions and \eqref{eq:info} are used in the numerical implementation.

For a strictly positive sample, $\ell(\theta)\to-\infty$ as $\theta\downarrow0$ or $\theta\uparrow\infty$. At the lower endpoint, each weighted component density tends to zero with a positive power of $\theta$. At the upper endpoint, the exponential decay dominates the polynomial prefactors. Continuity therefore ensures at least one interior likelihood maximizer. Uniqueness is not asserted.

For numerical estimation, evaluate the likelihood using the log-sum-exp identity
\[
 \log(A+B)=m+\log\{e^{\log A-m}+e^{\log B-m}\},\quad
 m=\max(\log A,\log B).
\]
Optimize in $\eta=\log\theta$, search a grid for candidate peaks, refine each candidate, and retain the largest likelihood. Inspect boundaries and expand the range if they are competitive. A numerical grid is a practical safeguard, not a proof of finding every finite-sample stationary point.

\subsection{A constrained expectation--maximization algorithm}
The latent representation \eqref{eq:latent} yields an expectation--maximization (EM) implementation in the sense of \citet{dempster1977}. Given $\theta^{(q)}$, calculate $r_i^{(q)}=r_i(\theta^{(q)})$ in the E step. Up to terms independent of $\theta$, the M-step objective is
\begin{align}\label{eq:EM}
 Q(\theta\mid\theta^{(q)})={}&\sum_i r_i^{(q)}
 \{\log w(\theta)+\log a(\theta)-a(\theta)X_i\}\notag\\
 &+\sum_i(1-r_i^{(q)})
 \{\log(1-w(\theta))+4\log b(\theta)-b(\theta)X_i\}.
\end{align}
Maximize this scalar function subject to the LR2 rate specification. Updating the weight and two rates independently would fit a different, unconstrained model. Exact maximization of $Q$ makes the observed likelihood nondecreasing, but EM can converge to a local stationary point. Multiple starts and comparison with the direct likelihood search remain useful. The simulation below uses direct likelihood optimization, not EM.

\subsection{Large-sample inference}
\begin{theorem}[Regular likelihood inference]\label{thm:asymptotic}
Let $\theta_0$ be an interior point of a fixed compact parameter interval $K\subset(0,\infty)$. For independent observations from LR2, a global maximizer over $K$ is consistent. Moreover,
\begin{equation}\label{eq:MLEasymptotic}
 \sqrt n(\widehat\theta_{\rm ML}-\theta_0)
 \xrightarrow{d}N(0,I(\theta_0)^{-1}),\qquad
 I(\theta)=\int_0^\infty\{\partial_\theta\log f(x;\theta)\}^2f(x;\theta)\dd x,
\end{equation}
where $0<I(\theta_0)<\infty$.
\end{theorem}
\begin{proof}
On $K$, weights and rates are bounded away from zero and infinity. The score and its first two parameter derivatives are bounded by polynomials in $x$, since posterior component probabilities lie between zero and one and rate derivatives are bounded on $K$. Component tails admit a uniform exponential envelope. The log density admits an integrable envelope of the form $C(1+x+|\log x|)$; in fact the positive exponential component controls the behavior near zero. These bounds give uniform laws of large numbers and justify differentiating under the integral. Identifiability (Proposition~\ref{prop:ident}) makes the expected log likelihood uniquely maximized at $\theta_0$ by the Kullback--Leibler inequality. Compactness then proves consistency. Finiteness of $I$ follows from the score bound and finite moments. If $I=0$, $\partial_\theta f(x;\theta_0)$ vanishes almost everywhere and, by continuity, everywhere on the positive support. Its continuous limit at zero would then vanish, contradicting the strictly positive derivative of \eqref{eq:origin}. Thus $I>0$. A score Taylor expansion at the consistent maximizer, the central limit theorem, and convergence of the normalized Hessian give \eqref{eq:MLEasymptotic}; these are the standard regular likelihood arguments described by \citet{vandervaart1998}.
\end{proof}

The compact-interval theorem does not establish a uniform approximation as $\theta_0$ approaches zero or infinity, nor does it prove that a local optimizer is a global maximizer. At a fitted solution with $J_n(\widehat\theta)>0$, an approximate $100(1-\alpha)\%$ interval that respects positivity is
\begin{equation}\label{eq:wald}
 \left[\exp\left\{\log\widehat\theta-
 \frac{z_{1-\alpha/2}}{\widehat\theta\sqrt{J_n(\widehat\theta)}}\right\},
 \exp\left\{\log\widehat\theta+
 \frac{z_{1-\alpha/2}}{\widehat\theta\sqrt{J_n(\widehat\theta)}}\right\}\right].
\end{equation}
Here $J_n$ is total information, so no additional factor of $n$ is inserted. A likelihood-based alternative is the set
\begin{equation}
 \{\theta>0:2[\ell(\widehat\theta)-\ell(\theta)]\leq\chi^2_{1,1-\alpha}\}.
\end{equation}
This cutoff concerns inference within LR2 under the stated regularity assumptions. Parametric bootstrap intervals can also be obtained by generating from the fitted model and refitting with the same search procedure.

For a differentiable functional $g(\theta)$, including a survival probability or the mean, the delta-method variance is $\{g'(\theta_0)\}^2/[nI(\theta_0)]$. A quantile derivative follows implicitly from $F(Q(p;\theta);\theta)=p$:
\begin{equation}
 \partial_\theta Q(p;\theta)=-\frac{\partial_\theta F(Q(p;\theta);\theta)}{f(Q(p;\theta);\theta)}.
\end{equation}

\subsection{Spacing estimation and independent right censoring}
For distinct, uncensored observations ordered as $X_{(1)}<\cdots<X_{(n)}$, let $F(X_{(0)})=0$ and $F(X_{(n+1)})=1$. The maximum product of spacings estimator maximizes
\begin{equation}\label{eq:spacing}
 \sum_{i=1}^{n+1}\log\{F(X_{(i)};\theta)-F(X_{(i-1)};\theta)\}.
\end{equation}
The general approach is due to \citet{cheng1983}. Stable differences of survival probabilities are useful in the upper tail. Ties caused by rounding produce zero spacings and require a grouped-data or tie-adjusted method; \eqref{eq:spacing} should not be applied unchanged. No numerical comparison of spacing estimation is claimed in this paper.

For independent right censoring, observe $T_i=\min(X_i,C_i)$ and $\delta_i=\mathbf1\{X_i\leq C_i\}$. Assuming the censoring mechanism is independent of $X_i$ and its distribution has no shared unknown parameter, the part of the log likelihood depending on $\theta$ is
\begin{equation}\label{eq:censored}
 \ell_c(\theta)=\sum_i\{\delta_i\log f(T_i;\theta)+(1-\delta_i)\log S(T_i;\theta)\}.
\end{equation}
The gamma-component posterior probability for a censored observation uses its survival contribution, not its density contribution. In particular, the exponential posterior probability is $we^{-aT_i}/S(T_i)$ for $\delta_i=0$. The complete-data formulation must accordingly include survival terms for censored records. Equation~\eqref{eq:EM}, written for complete lifetimes, is not the censored-data M-step. The uncensored moment estimator is likewise inappropriate when applied naively to $T_i$. Censored estimation is an extension described by its likelihood here and is not evaluated in the simulation.

\section{Numerical study and reproducibility}\label{sec:numerical}
\subsection{Design and implementation}
We report a Monte Carlo experiment for LR2 with $\theta\in\{0.3,0.5,1,2\}$ and $n\in\{50,200\}$. Each of the eight settings uses 500 independently generated samples, for 4,000 datasets in total. These results are the LR2 subset of the combined experiment from which this standalone analysis was developed, not a new independent experiment. The original random stream uses NumPy's \texttt{default\_rng} with seed 20261009. The accompanying script preserves the stream position used for this model before generating its samples. Indicators use $w(\theta)$, and lifetimes are drawn from the appropriate exponential or gamma component. No observed lifetime dataset was used.

For each sample, we computed the moment estimator by a bracketed root search on the log-parameter scale, and the likelihood estimator by evaluating 121 equally spaced log-parameter values on $[-9,9]$, refining all grid-local maxima with bounded scalar optimization. The retained candidate had the largest likelihood. Endpoint comparisons checked that a boundary was not preferred. This is a numerical approximation to global maximization with an explicitly documented search range.

We report empirical bias and root mean squared error (RMSE),
\begin{equation}
 \widehat{\mathrm{Bias}}=\frac1B\sum_{r=1}^B(\widehat\theta_r-\theta),\qquad
 \widehat{\mathrm{RMSE}}=\left\{\frac1B\sum_{r=1}^B(\widehat\theta_r-\theta)^2\right\}^{1/2},
 \quad B=500.
\end{equation}
Coverage refers to the nominal 95\% log-scale Wald interval in \eqref{eq:wald}. The machine-readable output also records optimizer failures, nonpositive observed-information counts, and Monte Carlo standard errors of coverage. An unavailable interval is recorded separately rather than silently treated as a valid result. With 500 replications, the standard error of a coverage estimate near 0.95 is about 0.010. Modest coverage differences should therefore be interpreted cautiously.

Before simulation, numerical quadrature verified normalization and the first two raw moments for LR2 at the four parameter values. Finite differences checked the analytic score and Hessian at three positive $x$ values per setting; differentiating the survival function numerically also checked the density. These checks support the implementation but do not substitute for the proofs.

\begin{table}[htbp]\centering\small
\caption{Monte Carlo bias, RMSE, and 95\% log-scale Wald coverage, based on 500 replications per setting. ML denotes maximum likelihood and MM the moment estimator.}\label{tab:simulation}
\begin{tabular}{crrrrrr}
\toprule
$\theta$ & $n$ & Bias (ML) & RMSE (ML) & Bias (MM) & RMSE (MM) & Coverage \\
\midrule
0.3 & 50 & 0.0000 & 0.0112 & -0.0000 & 0.0117 & 0.948 \\
0.3 & 200 & 0.0001 & 0.0053 & 0.0000 & 0.0055 & 0.970 \\
0.5 & 50 & 0.0025 & 0.0195 & 0.0028 & 0.0206 & 0.958 \\
0.5 & 200 & 0.0003 & 0.0096 & 0.0003 & 0.0103 & 0.964 \\
1 & 50 & 0.0078 & 0.0529 & 0.0079 & 0.0553 & 0.934 \\
1 & 200 & 0.0003 & 0.0245 & 0.0002 & 0.0263 & 0.958 \\
2 & 50 & 0.0182 & 0.2286 & 0.0314 & 0.2499 & 0.942 \\
2 & 200 & -0.0049 & 0.0990 & -0.0033 & 0.1060 & 0.950 \\
\bottomrule
\end{tabular}

\end{table}

\subsection{Results}
Table~\ref{tab:simulation} shows that increasing $n$ from 50 to 200 reduces RMSE at every reported parameter value for both estimators. Maximum likelihood has lower empirical RMSE than moment estimation in 8 of the eight settings. At $\theta=2$ and $n=50$, ML and MM RMSE are 0.2286 and 0.2499; at $n=200$ they are 0.0990 and 0.1060. These are within-model estimator comparisons, rather than evidence of superiority over competing lifetime distributions.

Observed 95\% interval coverage ranges from 0.934 to 0.970. The lowest reported coverage occurs at $\theta=1$ and $n=50$. With 500 replications per setting, Monte Carlo uncertainty is approximately one percentage point near nominal coverage. All 4,000 likelihood fits in this article returned successful optimization results, no boundary comparison required expansion, and all fitted observed-information values were positive. These finite-sample results support approximate inference within the reported design; they do not guarantee nominal coverage outside it.

\subsection{Scope of empirical comparison}
The simulation assesses estimation under correct specification. It does not compare predictive performance against the Lindley, Weibull, lognormal, or unrestricted exponential--gamma mixture distributions, and it is not evidence of empirical superiority. A future application should specify the data provenance, working time unit, censoring and rounding mechanisms, and sample exclusions before fitting. Positive observations alone are insufficient to validate a lifetime interpretation.

For observed data, a reasonable comparison would fit a constant-hazard exponential model, the Lindley model, flexible two-parameter lifetime models, and an unrestricted exponential--gamma mixture with shape four. For a regular candidate model with $k$ fitted parameters, the Akaike criterion is $\mathrm{AIC}=-2\ell(\widehat\theta)+2k$ \citep{akaike1974}. The fixed-unit LR2 model has $k=1$, so its criterion is $-2\ell(\widehat\theta)+2$. If an additional time scale is estimated, it must count as an additional parameter. These are comparison principles, not fitted results.

Probability plots and held-out log likelihood can supplement information criteria. A goodness-of-fit statistic evaluated after estimating parameters does not generally have the same null distribution as when parameters are known. For example, a parametric-bootstrap CDF test should simulate from the fitted model, refit each bootstrap sample, and recompute the statistic. A real-data results table is deliberately omitted because no application dataset has been analyzed.

\section{Discussion and conclusions}\label{sec:discussion}
LR2 is a one-dimensional restriction of an exponential--gamma mixture with gamma shape four. The common representation yields explicit distributional and reliability quantities, while the constraints permit unique moment estimation and identifiable likelihood inference for the specified model. The density-turning-point threshold and the hazard results give more informative shape statements than visual inspection alone. In particular, an increasing component hazard is not enough to justify describing LR2 as an increasing-hazard distribution.

For LR2, the crossing of the linear exponential rate and quadratic gamma rate changes the tail regime at $\theta=1/2$. The mean nevertheless remains strictly decreasing across this crossing. Thus a change in the component governing long durations does not prevent unique moment estimation.

Several restrictions matter in applications. First, the component rates and the mixture weight cannot vary independently. This can make a fit parsimonious, but can also force an unrealistic relationship between the frequency and duration of the two mechanisms. Second, the fixed shape four is a modeling choice, not a finding established by observed data. Third, the current numerical study is limited to complete, independent, correctly specified samples and a finite parameter grid. It does not establish robustness to censoring, dependence, rounding, or misspecification.

The time-unit issue deserves particular attention. If $Y=cX$, then its component rates become $a/c$ and $b/c$, while the mixture weight remains $w(\theta)$. Because $w$ is injective, representing $Y$ by the same one-parameter model requires the same $\theta$; that would leave both rates unchanged. Thus LR2 is not closed under arbitrary time-unit changes. The linked weight and rates are interpreted numerically in that working unit; the parameter does not independently adjust the time scale.

A scale extension resolves this practical restriction. If $X$ denotes a main-model lifetime in the reference unit, set $Y=\tau X$ with $\tau>0$. Then
\begin{equation}
 f_Y(y;\theta,\tau)=\tau^{-1}f_X(y/\tau;\theta),\qquad
 S_Y(y;\theta,\tau)=S_X(y/\tau;\theta).
\end{equation}
This is a two-parameter model with rates $a(\theta)/\tau$ and $b(\theta)/\tau$. Under $Y^*=cY$, only $\tau$ changes to $c\tau$. Estimation and identifiability of this extension require their own analysis and are not covered by the one-parameter inference theorem.

Further work could evaluate this scale extension, relax the fixed gamma shape, and assess censoring and empirical performance with appropriately documented datasets. A full classification of hazard turning points would also strengthen the shape analysis. The present results establish mathematical and computational foundations for this particular restriction, while keeping the distinction between analytic properties and demonstrated application performance explicit.

\section*{Data and code availability}
All numerical results use generated data. The accompanying Python script \texttt{reproduce.py} regenerates both figures, the Monte Carlo table, the numerical-check report, and software-version metadata. The source bundle includes the CSV output and environment record. No confidential or proprietary data are involved.

\appendix
\section{Exact simulation algorithm and software details}
For LR2 at a specified parameter, generate independent random variables $U\sim\operatorname{Uniform}(0,1)$, $E\sim\Exp(a)$, and $G\sim\Gam(4,b)$, and return
\begin{equation}
 X=E\mathbf1\{U<w\}+G\mathbf1\{U\geq w\}.
\end{equation}
This is exact mixture sampling and requires no inverse-CDF approximation. The gamma variable can alternatively be obtained as the sum of four independent exponentials of rate $b$. The provided implementation draws all three vectors for each replication so that the order of random-number consumption is fixed.

The Python dependencies are NumPy, SciPy, and Matplotlib. The actual versions used for the reported experiment are recorded in \texttt{environment.json}. Running \texttt{python3 anc/reproduce.py --replications 500} regenerates the outputs. The full likelihood and observed-information formulas are implemented rather than approximating information with a numerically unstable second difference. Numerical finite differences are used only for independent implementation checks.

\bibliographystyle{plainnat}
\bibliography{references}

@article{lindley1958,
 author={Lindley, Dennis V.}, title={Fiducial Distributions and {Bayes}' Theorem},
 journal={Journal of the Royal Statistical Society: Series B (Methodological)},
 year={1958}, volume={20}, number={1}, pages={102--107}, doi={10.1111/j.2517-6161.1958.tb00278.x}}

@article{ghitany2008,
 author={Ghitany, M. E. and Atieh, B. and Nadarajah, S.},
 title={Lindley distribution and its application}, journal={Mathematics and Computers in Simulation},
 year={2008}, volume={78}, number={4}, pages={493--506}, doi={10.1016/j.matcom.2007.06.007}}

@article{shanker2016,
 author={Shanker, Rama}, title={Sujatha distribution and its applications},
 journal={Statistics in Transition New Series}, year={2016},volume={17},number={3},pages={391--410},
 url={https://sit.stat.gov.pl/Article/435}}

@article{mclachlan2019,
 author={McLachlan, Geoffrey J. and Lee, Sharon X. and Rathnayake, Suren I.},
 title={Finite Mixture Models},journal={Annual Review of Statistics and Its Application},
 year={2019},volume={6},pages={355--378},doi={10.1146/annurev-statistics-031017-100325}}

@article{sakthivel2023,
 author={Sakthivel, K. M. and {Vidhya G}},
 title={A New Finite Mixture of Probability Models with Application},
 journal={Reliability: Theory \& Applications},year={2023},volume={18},number={3},pages={135--153},
 url={https://gnedenko.net/Journal/2023/032023/RTA_3_2023-09.pdf}}

@article{dempster1977,
 author={Dempster, A. P. and Laird, N. M. and Rubin, D. B.},
 title={Maximum Likelihood from Incomplete Data via the {EM} Algorithm},
 journal={Journal of the Royal Statistical Society: Series B (Methodological)},
 year={1977},volume={39},number={1},pages={1--22},doi={10.1111/j.2517-6161.1977.tb01600.x}}

@article{cheng1983,
 author={Cheng, R. C. H. and Amin, N. A. K.},
 title={Estimating Parameters in Continuous Univariate Distributions with a Shifted Origin},
 journal={Journal of the Royal Statistical Society: Series B (Methodological)},
 year={1983},volume={45},number={3},pages={394--403},doi={10.1111/j.2517-6161.1983.tb01268.x}}

@book{vandervaart1998,
 author={van der Vaart, A. W.}, title={Asymptotic Statistics},
 publisher={Cambridge University Press},address={Cambridge},year={1998},doi={10.1017/CBO9780511802256}}

@book{cox1984,
 author={Cox, D. R. and Oakes, D.},title={Analysis of Survival Data},
 publisher={Chapman and Hall},address={London},year={1984}}

@article{akaike1974,
 author={Akaike, Hirotugu},title={A new look at the statistical model identification},
 journal={IEEE Transactions on Automatic Control},year={1974},volume={19},number={6},
 pages={716--723},doi={10.1109/TAC.1974.1100705}}

@article{ghitany2013,
 author={Ghitany, M. E. and Al-Mutairi, D. K. and Balakrishnan, N. and Al-Enezi, L. J.},
 title={Power Lindley distribution and associated inference},
 journal={Computational Statistics \& Data Analysis},year={2013},volume={64},pages={20--33},
 doi={10.1016/j.csda.2013.02.026}}

@article{sen2016,
 author={Sen, Subhradev and Maiti, Sudhansu S. and Chandra, N.},
 title={The xgamma Distribution: Statistical Properties and Application},
 journal={Journal of Modern Applied Statistical Methods},year={2016},volume={15},number={1},pages={774--788},
 doi={10.22237/jmasm/1462077420}}
\end{document}